\documentclass[11pt,pdf]{article}
\ifx\pdfoutput\undefined
\else
\fi
\ifx\pdftexversion\undefined
\usepackage{amssymb,graphicx,epsfig,epstopdf}
\else
\usepackage[pdftex]{graphicx,epsfig}
\fi

\usepackage{url}
\usepackage[bookmarks=false,colorlinks=false,colorlinks={true}, citecolor={black}, filecolor={black}, linkcolor={black}, urlcolor={black},pdfstartview={XYZ null null 1.22}]{hyperref}
\usepackage[sort&compress,numbers]{natbib}
\usepackage{amsmath,amsthm,graphicx,epstopdf,amssymb,fullpage}
\usepackage{algorithm}
\usepackage{algorithmic}
\usepackage{color}
 
\newcommand{\Section}[1]{\section{#1}}

\newtheorem{theorem}{Theorem}[section]
\newtheorem{definition}{Definition}[section]

\newtheorem{claim}[theorem]{Claim}

\newtheorem{corollary}[theorem]{Corollary}

\begin{document}
\begin{titlepage}
\begin{center} 
{\Large {\sc Data Interpolation}\\ 
\vspace{0.3cm}
An Efficient Sampling Alternative for Big Data Aggregation}\\
\vspace{1.3cm}
by\\
\vspace{1.3cm}
\textbf{Hadassa Daltrophe, Shlomi Dolev and Zvi Lotker}\\
\vspace{1.3cm}
Technical Report \#13-01\\
\vspace{1.3cm}
September 2012
\end{center}
\end{titlepage}

\renewcommand{\thefootnote}{\fnsymbol{footnote}}
\title{{\sc Data Interpolation}\\ 
An Efficient Sampling Alternative for Big Data Aggregation\\
{\large (Technical Report)} }

\author{
Hadassa Daltrophe\footnotemark[1] 
\and 
Shlomi Dolev\footnotemark[2]
\and  
Zvi Lotker\footnotemark[3]
}
\maketitle

\date{}

\footnotetext[1]{
Department of Computer Science,
Ben-Gurion University, Beer-Sheva, 84105, Israel.
Email: {\tt{hd}@cs.bgu.ac.il}.
 Partially supported by a Russian Israeli grant from the Israeli
Ministry of Science and Technology and the Russian Foundation for
Basic Research, the Rita Altura Trust Chair in Computer Sciences, the
Lynne and William Frankel Center for Computer Sciences, Israel Science
Foundation (grant number 428/11), Cabarnit Cyber Security MAGNET
Consortium, Grant from the Institute for Future Defense Technologies
Research named for the Medvedi of the Technion, MAFAT, and Israeli
Internet Association
}
\footnotetext[2]{
Department of Computer Science,
Ben-Gurion University, Beer-Sheva, 84105, Israel.
Email:{\tt{dolev}@cs.bgu.ac.il}.
Partially supported by Deutsche Telekom, Rita Altura Trust Chair in
Computer Sciences, Lynne and William
Frankel Center for Computer Sciences, Israel Science Foundation (grant
number 428/11) and Cabarnit Cyber
Security MAGNET Consortium.
}
\footnotetext[3]{
Department of Communication Systems Engineering,
Ben-Gurion University, Beer-Sheva, 84105, Israel.
Email: {\tt{zvilo}@bgu.ac.il}.
}

\renewcommand{\thefootnote}{\arabic{footnote}}
\thispagestyle{empty}
\begin{abstract}
Given a large set of measurement sensor data, in order to identify a simple function that
captures the essence of the data gathered by the sensors, we suggest representing the data by
(spatial) functions, in particular by polynomials. Given a (sampled) set of values, we interpolate
the datapoints to define a polynomial that would represent the data. The interpolation is
challenging, since in practice the data can be noisy and even Byzantine, where the Byzantine
data represents an adversarial value that is not limited to being close to the correct measured
data. We present two solutions, one that extends the Welch-Berlekamp technique in the case of
multidimensional data, and copes with discrete noise and Byzantine data, and the other based on Arora and Khot techniques, extending them in the case of multidimensional noisy
and Byzantine data.
\end{abstract}


\newpage
\setcounter{page}{1}
\pagenumbering{arabic}
\Section{Introduction}
\label{s:Introduction}
Consider the task of representing information in an error-tolerant way, such that it can be introduced even if it contains noise or even if the data is partially corrupted and destroyed. Polynomials are a common venue for such approximation, where the goal is
to find a polynomial $p$ of degree at most $d$ that would
represent the entire data correctly. 

Our motivation comes from sensor data aggregation, and the need to extend the distributed aggregation to distributed interpolation, use sampling to cope with huge data and anticipate the value of missing data. For example, a sensor network may interact with the physical environment, while each node in the network is may sense the surrounding environment (e.g., temperature, humidity etc). The environmental measured values should be transmitted to a remote repository or remote server. Note that the environmental values usually contain noise, and
there can be malicious inputs, i.e., part of the data may be corrupted. 

In contrast to distributed data aggregation where the resulting computation is a function such as COUNT, SUM and AVERAGE (e.g.~\cite{Aggregation2006,Aggregation2007,Aggregation2011}),
in distributed data interpolation, our goal is to represent every value of the data by a single (abstracting) function. Our computational model consists of sampling the sensor network data and estimating the missing information using polynomial manipulations.

The management of big data systems also gives motivation for the distributed interpolation
method. The abstraction of big data becomes one of the most important tasks in the presence of
the enormous amount of data produced these days. Communicating and analyzing the entire data
does not scale, even when data aggregation techniques are used. This study suggests a method to
represent the distributing big data by a simple abstract function (such as polynomial) which will lead
to effective use of that data.

We suggest interpolating the big data in the scope of distributed systems by using local {\it data centers}. Each data center samples the data around it and computes a polynomial that reflects the local data. The local polynomials are merged to a global one by interpolation in a hierarchical manner. In the process of calculating the local polynomials noise and Byzantine data samples are eliminated.



\subsection*{Basic Definitions.}

\begin{itemize}
\item For {\it multivariate polynomial} $p(\textbf{\texttt{x}})\in \mathbb{R}[\textbf{\texttt{x}}]=\mathbb{R}[x_1,...,x_k]$ let $\left\|p\right\|_\infty =sup\left\{\left|p(x_1,...,x_k)\right|:x_1,...,x_k\in \mathbb{R}\right\}$.

\item A {\it monomial} in a collection of variable $x_i,...,x_n$ is a product
\begin{equation}
	x_1^{\alpha_1}x_2^{\alpha_2}x_n^{\alpha_n} \nonumber
\end{equation}
where $\alpha_i$ are non-negative integers.

\item The {\it total degree} of a multivariate polynomial $p$ is the maximum degree of any term in $p$, where the degree of particular term is the sum of the variable exponents.

\item A polynomial $q$ is a $\delta${\it-approximation} to $p$ if $\left\|p-q\right\|_\infty \leq \delta$.

\end{itemize}

\subsection*{Polynomial Fitting to Noisy and Byzantine Data.}
Formally, in this paper, we learn the following problem:

\begin{definition}[Polynomial Fitting to Noisy and Byzantine Data Problem]
\label{d:mainProb}
Given a sample $S$ of $k$ dimension datapoints $\left\{\left(x_{1_i},...,x_{k_i}\right)\right\}^{N}_{i=1}$ and a function $f$ defined on those points $f(x_{1_i},...,x_{k_i})=y_i$, a noise parameter $\delta>0$ and a threshold $\rho>0$, we have to find a polynomial $p$ of total degree $d$ satisfying 
\begin{eqnarray}
\label{eq:mainDef}
p(x_{1},...,x_{k}) \in \left[y-\delta,y+\delta\right] 
 \text{ for at least $\rho$ fraction of $S$}
\end{eqnarray}
\end{definition}

Generally, we propose the use of polynomials to represent large amounts of sensor data. The process works by sampling the data and then using this sample to construct a polynomial whose distance (according to the $\ell_\infty$ metric) from the polynomial constructed using the \textit{whole} data set is small. The main challenges to this approach are $(i)$ the presence of noise (identified by the $\delta$ parameter), and $(ii)$ arbitrarily corrupted data (Byzantine data, denoted by $\rho$) that can cause inaccurate sampling and, thus, lead to badly constructed polynomials.
 
Given that the function $f$ is continuous, by the Weierstrass approximation Theorem~\cite{Weierstrass} we know that for any given $\epsilon>0$, there exists a polynomial $p'$ such that 
\begin{equation}
	\left\|f-p'\right\|_\infty<\epsilon
\end{equation}
This can tell us that our desired polynomial $p$ exists (i.e., $p'=p$ and $\epsilon=\delta$, satisfying eq.\ref{eq:mainDef}), and we can relate the data as arising from polynomial function (i.e., the unknown function $f$ is $d$ degree polynomial we need to reconstruct), and this is the underlying model assumed in the paper.

One obvious candidate to construct approximating polynomial is interpolation at equidistant points. However, the sequence of interpolation polynomials does \textit{not} converge uniformly to $f$ for all $f \in C[0, 1]$ due to Runge's phenomenon~\cite{Davis1975}. Chebyshev interpolation (i.e., interpolate $f$ using the points defined by the Chebyshev polynomial) minimizes Runge's oscillation, but it is not suffice the polynomial fitting problem presented above (Definition~\ref{d:mainProb}) due to the randomly distributed data we have assumed.

Taylor polynomials are also not appropriate; for even setting aside questions of convergence, they are applicable only to
functions that are \textit{infinitely differentiable}, and not to all continuous functions.
 
Another classical polynomial sequence is suggested by S. Bernstein~\cite{B1912} as constructive proof of the Weierstrass Theorem. Bernstein polynomial: $B^{f}_{n}(x)=\displaystyle\sum^{n}_{i=0}f\left(\frac{i}{n}\right){n\choose i} x^i(1-x)^{n-i}$ converges uniformly to any continuous function $f$ which is bounded on $[0,1]$. The formal Berenstein polynomial samples the function $f$ in an \textit{equidistant} fashion. To handle a random sample data, we can use Vitale~\cite{Vitale1975} results which consider that the datapoints $S=x_1,...,x_N$ are \textit{i.i.d} observations drawn from an unknown \textit{density} function $f$. The \textit{Bernstein polynomial estimate of $f$} defined as 
$\tilde{B}^{f}_{n}(x) = \frac{n+1}{N} \displaystyle\sum^{n}_{i=0}\mu_{in}^N{n\choose i} x^i(1-x)^{n-i}$
where $\mu_{in}^N$ is the number of points ($x_i$'s) appear in the interval $[\frac{i}{n+1},\frac{i+1}{n+1}]$. 
Vitale~\cite{Vitale1975} showed that $\left\|\tilde{B}^{f}_{n}(x) - f\right\|_\infty\leq \epsilon$ for every given $\epsilon > 0$.

Tenbusch~\cite{Tenbusch1994} extended Vitale's idea to multidimensional densities, where there is need to note that those works hold only when the datapoints are \textit{i.i.d} observations. Another reason not to use the Bernstien polynomial is the slow convergence rate (Voronovskaya's Theorem states that for functions that are twice differentiable, the rate of convergence is about $1/n$, see Davis~\cite{Davis1975}). 

Considering other classical curve-fitting and approximation theories~\cite{Rivlin1969}, most research has used the $\ell_2$ norm of noise, such as the method of least square errors. These attitudes not suffice the adversarial noise we have assumed here. To our knowledge, only~\cite{AroraKhot2002} referred the $\ell_\infty$ noise that fits our considered problem and we further relate~\cite{AroraKhot2002} study. 

The polynomial fitting problem as stated in Definition~\ref{d:mainProb} can also be studied by Error-Correcting Code Theory. From that point of view, extensive literature exists dealing with the
noise-free case (i.e., $\delta = 0$ and $\rho<1$). In the next section, we present an algorithm that handles a combination of discrete noise and Byzantine data based on the Welch-Berlekamp~\cite{WelchBerlekamp1986} error-elimination method. Moreover, the fundamental Welch-Berlekamp algorithm treats only the one-dimension case, where we suggest a means to deal with corrupted-noisy data appearing at one and multi-dimensional inputs.

Related to unrestricted noise, we refer to the polynomial-fitting problem as defined by
Arora and Khot~\cite{AroraKhot2002}. Based on their results in Section~\ref{s:MultidimensionalDataFitting}, we introduce the polynomial fit generalization, where we provide a polynomial-time algorithm dealing with multi-variate data.

Summarizing, this work provides the following contributions:
\begin{itemize}
\item We describe an algorithm that constructs a polynomial using the Welch-Berlekamp (WB) method as a subroutine. The algorithm is tolerated to discrete-noise and Byzantine data.
 
\item  We identify how the previous method can be generalized to handle multi-dimensional data. Moreover, we present a multivariate analogue of the WB method, under conditions which will be specified.

\item Using linear programing minimization and the Markov-Bernstein Theorem, we generalized Arora and Khot algorithm to reconstruct an unknown \textit{multi-dimensional} polynomial. Furthermore, we detail the way to eliminate the Byzantine appearance when such inputs exist.

\end{itemize}

Those three points stated in the three algorithms presented in the paper. The first Algorithm handles one-dimensional Byzantine data that contains discrete-noise. Algorithm~\ref{algo:2} generalized the WB idea to deal with multivariate malicious data. Finally, Algorithm~\ref{algo:3} summarized our approach to cope with unrestricted noise appeared in the (partially corrupted) data.


\Section{Discrete Finite Noise}
\label{s:DiscreteFiniteNoise}
In this section, we will study a simple aspect of the polynomial fitting problem posed in Definition~\ref{d:mainProb}, where the data function is a polynomial, and we relaxed the noise constraint to be finite and discrete, i.e., the noise $\delta$ is defined on a finite field $\mathbb{F}_q$ containing $q$ elements.

Welch and Berlekamp related the problem of polynomial reconstruction in their decoding algorithm for Reed Solomon codes~\cite{WelchBerlekamp1986}. The main idea of their algorithm is to describe the received (partially corrupted) data as a ratio of polynomials. Their solution holds for noise-free cases and a limited fraction of the corrupted data ($\delta=0,\rho>1/2$). Almost $30$ years later, Sudan's list decoding algorithm~\cite{Sudan1997} relaxed the Byzantine constraint ($\delta=0,\rho$ can be less than $1/2$) by using bivariate polynomial interpolation. Those concepts do not hold up well in the noisy case since they use the roots of the polynomial and the divisibility of one polynomial by other methods that are problematic for noisy data (as shown in~\cite{AroraKhot2002}, Section $1.2$). Here, we will use the WB algorithm~\cite{WelchBerlekamp1986} as a ``black box" to obtain an algorithm that handles the discrete-noise notation of the polynomial-fitting problem.

Given a data set $\left\{\left(x_i,y_i\right)\right\}^{N}_{i=1}$ that is within a distance of $t=\rho N$ from some polynomial $p(x)$ of degree $< d$, the WB approach to eliminate the irrelevant data is to use the roots of an object called \textit{the error-locating polynomial}, $e$. In other words, we want $e(x_i)=0$ whenever $p(x_i)\neq y_i$. This is done by defining another polynomial
$q(x)\equiv p(x)e(x)$. To resolve these polynomials we need to solve the linear system, $q(x_i)=y_ie(x_i) \text{ for all $i$.}$

Welch and Berlekamp show that $e(x)|q(x)$ and $p(x)$ can be found by the ratio $p(x)=q(x)/e(x)$ at $O(N^\omega)$ running time (where $\omega$ is the matrix multiplication complexity). In Algorithm~\ref{algo:1}, we are use the WB method as a subroutine to manage the noisy-corrupted data.

\begin{algorithm}
\caption{Reconstruct the polynomial $p(x)$ representing the true data }
\begin{algorithmic}
\label{algo:1}
\REQUIRE $S,S'\subseteq S,\rho,d,\delta,\Delta={v_1,...,v_{|S'|}}$
\STATE $i \leftarrow 0$
\REPEAT
\STATE $i \leftarrow i+1$
\STATE $S'_i \leftarrow  S'+v_i $
\STATE $p_i(x) \leftarrow WB(S'_i,d,\rho)$
\UNTIL{$p_i(x_j)\in [y_j-\delta,y_j+\delta]$ for at least $\rho$ fraction of $j$'s; $(x_j,y_j) \in S-S'_i$ }
\RETURN $p(x)\leftarrow p_i(x)$
\end{algorithmic}
\end{algorithm}

Given \textit{any} sample $S$ such that $\rho$ fraction of $S$ is not corrupted, we will choose a subset $S'\subseteq S$ in a size related to the desired degree $d$ and $\rho$ (the WB algorithm requires $2t+d$ points, where $t=\rho N$ is the number of the corrupted points). At every step $i$, we will add $S'$ different values of noise as defined by the set $\Delta$ which contain all the vectors of length $|S'|$ assigned the elements of $\mathbb{F}_q$ in lexicographic order, i.e., $\Delta=\left\{(a_1,...,a_{|S'|}):a_i\in \mathbb{F}_q\right\}$. Now, we can reconstruct the polynomial $p_i$ using the WB algorithm. The resulting polynomial $p_i$ is tested by the original dataset $S$, where the criteria is that $p_i$ is within $\delta$ from all nodes but the Byzantine nodes (according to the maximal number of Byzantine as defined by $\rho$). 

Since we assume a discrete finite noise ($\delta \in \mathbb{F}_q$), for each datapoint at the subset $S'$ (of size $O(d+\rho|S|)$), there is a possibly of $q$ values (where $q$ is a constant). Thus, in the worst case, when we run the WB polynomial algorithm for every possible value, it will cost $poly(d+N)$ time.

Note that if the desired polynomial's degree $d$ is not given, we can search for the minimal degree of a polynomial that fits the $\delta$ and number of Byzantine node restrictions in a binary search fashion. 

\subsection*{Multidimensional Data.}
To generalize the former algorithm to handle multidimensional data, there is need to formalize the WB algorithm to deal with multivariate polynomials. This is a challenging task due to the {\it infinite} roots those polynomials may have (and as previously mentioned, the WB method is strictly based on the polynomials' roots). 

A suggested method to handle $3$-dimensional data is to assume that the values of datapoints in one direction (e.g., x-direction) are distinct. This can be achieved by assuming the inputs $S=(x_1,y_1,f(x_1,y_1))...,(x_N,y_N,f(x_N,y_N))$ are \textit{i.i.d} observations. Moreover, we allow the malicious authority to change the observation input but not its distribution (i.e., to determine $z_i=f(x_i,y_i)$ value only). This assumption forces the data to have different $x_i$'s values, which help us to define the error locating polynomial $e$ in the x-direction only (or symmetrically over the y-axis). 

The 3-dimensional polynomial reconstruction is described in Algorithm~\ref{algo:2}.
\begin{algorithm}
\caption{Reconstruct the polynomial $p(x,y)$ representing the true data }
\label{algo:2}
\begin{itemize}
\item \textbf{Input}: $0<t=\rho N$ which is the Byzantine appearance bound, the total degree $d>1$ of the goal polynomial and $N$ triples ${(x_i,y_i,z_i)}_{i=1}^N$ with distinct $x_i$'s .
\item \textbf{Output}: Polynomial $p(x,y)$ of total degree at most $d$ or fail.
\item \textbf{Step 1}: Compute a non-zero univariate polynomial $e(x)$ of degree exactly $t$ and a bivariate polynomial $q(x,y)$ of total degree $d+t$ such that:
\begin{eqnarray}
\label{eq:bw3}
	z_ie(x_i)= q(x_i,y_i) & 1\leq i \leq N
\end{eqnarray} 
If such polynomials do not exist, output fail.
\item \textbf{Step 2}: If $e$ does not divide $q$, output fail, else compute $p(x,y) = \displaystyle\frac{q(x,y)}{e(x)}$. 
If $\Delta(z_i,p(x_i,y_i)_i)>t$, output fail. else output $p(x,y)$.
\end{itemize}
\end{algorithm}

\begin{theorem}
\label{t:3DBW}
Let $p$ be an unknown $d$ total degree polynomial with two variables. Given a threshold $\rho>0$ and a sample $S$ of $N={d+t+m \choose d+m}+t$ ($t=\rho N$) random points ${(x_i,y_i,z_i)}_{i=1}^N$ such that $$z_i=p(x_i,y_i)\text{ for at least $\rho$ fraction of $S$.}$$ The algorithm above reconstructs $p$ at $O(N^\omega)$ running time (where $\omega$ is the matrix multiplication complexity). 
\end{theorem}

\begin{proof}

The proof of the Theorem above follows from the subsequent claims.

\begin{claim}[Correctness]
\label{c:WB1}
There exist a pair of polynomials $e(x)$ and $q(x,y)$ that satisfy \textbf{Step 1} such that $q(x,y)=p(x,y)e(x)$.
\end{claim}
\begin{proof}
Taking the error locator polynomial $e(x)$ and $q(x,y)=p(x,y)e(x)$, where $deg(q) \leq deg(p) + deg(e) \leq t + d$. By definition, $e(x)$ is a degree $t$ polynomial with the following property:
$$
e(x) = 0 \text{ iff } z_i\neq p(x,y)
$$
We now argue that $e(x)$ and $q(x,y)$ satisfy eq.~\ref{eq:bw3}. Note that if $e(x_i)=0$, then $q(x_i,y_i)=z_ie(x_i)=0$. When $e(x_i)\neq 0$, we know $p(x_i,y_i)=z_i$ and so we still have $p(x_i,y_i)e(x_i)=z_ie(x_i)$, as desired.
\end{proof}

\begin{claim}[Uniqueness]
\label{c:WB2}
If any two distinct solutions $(q_1(x,y),e_1(x))\neq(q_2(x,y),e_2(x))$ satisfy\textbf{ Step 1}, then
they will satisfy $\displaystyle\frac{q_1(x,y)}{e_1(x)}=\frac{q_2(x,y)}{e_2(x)}$.
\end{claim}
\begin{proof}
It suffices us to prove that $q_1(x,y)e_2(x)=(q_2(x,y)e_1(x)$. Multiply this with $z_i$ and substitute $x,y$ with $x_i,y_i$, respectively,
$$
q_1(x_i,y_i)e_2(x_i)z_i=q_2(x_i,y_i)e_1(x_i)z_i
$$
We know, $\forall i \in [N]$ $q_1(x_i,y_i)=e_1(x_i)z_i$ and $q_2(x_i,y_i)=e_2(x_i)z_i$
If $z_i=0$, then we are done. Otherwise, if $z_i\neq 0$, then $q_1(x_i,y_i)=0, q_(x_i,y_i)=0\Rightarrow q_1(x,y)e_2(x)=(q_2(x,y)e_1(x)$ as desired.
\end{proof}

\begin{claim}[Time complexity]
Given $N=t+{d+t+2 \choose d+t}$ data samples, we can reconstruct $p(x,y)$ using $O(N^\omega)$ running time.
\end{claim}
\begin{proof}
\renewcommand{\qedsymbol}{}
Generally, for $m$ variate polynomial with degree $d$, there are ${d+m \choose d}$ terms~\cite{Saniee2008}; thus, it is a necessary condition that we have $t+{d+t+2 \choose d+t}$ distinct points for $q$ and $e$ to be uniquely defined. We have $N$ linear equation in at most $N$ variables, which we can solve e.g., by Gaussian elimination in time $O(N^\omega)$ (where $\omega$ is the matrix multiplication complexity). 

Finally, \textbf{Step 2} can be implemented in time $O(NlogN)$ by "long division"~\cite{Alfred1974}. Note that the general problem of deciding whether one multivariate polynomial divides another is related to computational algebraic geometry (specifically, this can be done using the Gr\"{o}bner base). However, since the divider is a univariate polynomial, we can mimic long division, where we consider $x$ to be the ``variable" and $y$ to just be some ``number."
\end{proof}
\end{proof}

\textbf{Example 1.} Suppose the unknown polynomial is $p(x,y)=x+y$. Given the parameters: $d=1$ (degree of $p$), $m=2$ (number of variable at $p$) and $t=1$ (number of corrupted inputs) and the set of $t+{d+t+2 \choose d+t}=7$ points:
$$
\text{\small{(1,2,2),(2,2,4),(6,1,7),(4,3,7),(8,2,0),(9,1,10),(3,7,10)}}
$$
that lie on $z=p(x,y)$. Following the algorithm, we define: $deg(e)=1,deg(q)=2$ and 
\begin{equation}
q_i =
\alpha_1x_i^2+\alpha_2x_iy_i+\alpha_3y_i^2+\alpha_4x_i+\alpha_5y_i+\alpha_{6} = z_i(x_i+\alpha_{7}) \nonumber 
\end{equation}
for coefficients $\alpha_1,...,\alpha_6,\beta$ and $1 \leq i \leq 12$. Note that we force $e(x)$ not to be the zero polynomial by define it to be monic (i.e., the leading coefficient equals to 1). Thus, we derive the linear system:
\begin{eqnarray}
\alpha_1 + \alpha_2 + \alpha_3 + \alpha_4 + \alpha_5 + \alpha_6 = 2\beta + 2 \nonumber\\
4\alpha_1 + 4\alpha_2 + 4\alpha_3 + 2\alpha_4 + 2\alpha_5 + \alpha_6 = 4\beta + 8 \nonumber\\
36\alpha_1 + 6\alpha_2 + \alpha_3 + 6\alpha_4 + \alpha_5 + \alpha_6 = 7\beta + 42 \nonumber\\
16\alpha_1 + 12\alpha_2 + 9\alpha_3 + 4\alpha_4 + 3\alpha_5 + \alpha_6 = 7\beta + 28\nonumber\\
64\alpha_1 + 16\alpha_2 + 4\alpha_3 + 8\alpha_4 + 2\alpha_5 + \alpha_6 = 0\nonumber\\
81\alpha_1 + 9\alpha_2 + \alpha_3 + 9\alpha_4 + \alpha_5 + \alpha_6 = 10\beta + 90\nonumber\\
9\alpha_1 + 21\alpha_2 + 49\alpha_3 + 3\alpha_4 + 7\alpha_5 + \alpha_6 = 10\beta + 30\nonumber
\end{eqnarray}
Solving the system gives: $q(x,y)= x^2+xy-8x-8y$ and $e(x)=x-8$. Dividing those polynomials, we get the expected solution: $q(x,y)/e(x)=p(x)=x+y$.

\begin{corollary}[Multivariate Polynomial Reconstruction]
Let $p$ be an unknown $d$ total degree polynomial with $m$ variable. Given a threshold $\rho>0$, a noise parameter $\delta$ and a sample $S$ of $N$ random points ${(x_{1_i},...,x_{m_i},y_i)}_{i=1}^N$ such that 
$$
y_i \in [p(x_{1_i},...,x_{m_i})-\delta,p(x_{1_i},...,x_{m_i})+\delta]\text{ for at least $\rho$ fraction of $S$}
$$
$p$ can be reconstructed using $N={d+m+\rho m \choose d+m}+\rho m$ datapoints.
\end{corollary}
\begin{proof}
Following Theorem~\ref{t:3DBW}, we can rewrite its proof for the multidimensional generalization. An interesting question is if there is any advantage to define the error-location polynomial $e$ to be multi-variate (instead of univariate as we previously presented). One ``advantage'' may be decreasing the required number of datapoints or improvement of the complexity. The size of the input data is strictly defined by the given bound on the corrupted data ($t=\rho N$) and the goal polynomial degree ($d$). Thus, there is no sense if the unknown coefficient comes from the highly \textit{degree} polynomial or from the highly \textit{dimension} polynomial, i.e., both options require the same size of sample, as illustrated in the Appendix.

Related to the complexity change, when the error-locating polynomial is multivariate, \textbf{Step 2} of Algorithm~\ref{algo:2} is more challenging since it contains multivariate polynomial division. A related reference is~\cite{Faugère1999} which is the most efficient implementation for the computation of Gr\"{o}bner bases relies on linear algebra. Using Gr\"{o}bner bases we can implement the division at close to $O(NlogN)$ time, as done in Algorithm~\ref{algo:2}.

To finish the proof, there is the need to explain how to deal with the noise. Since we assume only discrete noise, we can dismiss it using the method illustrated at Algorithm~\ref{algo:1}. We consistently insert a vector of possible noise and try to reconstruct the polynomial using Algorithm~\ref{algo:2}.
\end{proof}

\Section{Random Sample with Unrestricted Noise}
\label{s:MultidimensionalDataFitting}

Motivated by applications in vision, Arora and Khot~\cite{AroraKhot2002} studied the univariate polynomial fitting to noisy data using $O(d^2)$ datapoints, where $d$ is the polynomial degree. In this part, we generalized their results to $k$-dimensional data. 

Since our motivation comes from sensor planar aggregation, we will focus on bivariate polynomial reconstruction, where the multivariate proof is symmetric. We assume by rescaling the data that each $x_i,y_i,f(x_i,y_i) \in [-1,1]$. Allowing small noise at every point and large noise occasionally then there may be too many polynomials agreeing with the given data. Thus, given the noise parameter $\delta$, our goal is to find a polynomial $p$ that is a $\delta$-approximation of $f$, i.e., $p$ is $\delta$-close in $\ell_\infty$ norm to the unknown polynomial.

Let $I$ be a set of $d^5$ equally spaced points that cover the interval $[-1,1]$. Given the random sample $S \subset I, |S|=\frac{d^2}{\delta}log(\frac{d}{\delta})$, we approach the reconstruction problem by defining a linear programming system with the fitting polynomial as its solution. To incorporate the constraint that the unknown polynomial must take values of $[-1,1]$, we move to Chebyshev's representation of the polynomial. Thus, each of its coefficients is at most $\sqrt{2}$ (see eq.~\ref{eq:lp2}). We represent Chebyshev's polynomial by $T_i(\cdot),T_j(\cdot)$, and the variables $c_{ij}$ at the system are the Chebyshev coefficients. We construct the LP:
\begin{center}
\begin{eqnarray}
& \text{minimize } \delta  \nonumber\\
s.t. & \nonumber
\end{eqnarray}
\begin{eqnarray}
\label{eq:lp1}
 f(x_k,y_k)-\delta \leq  \sum_i^{n}{\sum_j^{m}{c_{ij}T_i(x_k)T_j(y_k)}} \leq  f(x_k,y_k)+\delta,
 & k=1,...,|S|  
 \end{eqnarray}
 \begin{eqnarray}
 \label{eq:lp2}
  |c_{ij}| \leq \sqrt{2}, & {i=1,...,n ; j=1,...,m} 
  \end{eqnarray}
  \begin{eqnarray}
  \label{eq:lp3}
  |\sum_i^{n}{\sum_j^{m}{c_{ij}T_i(x)T_j(y)}} | \leq 1, & \forall_{x,y} \in I
  \end{eqnarray}
\end{center}

The following Theorem presents our main result for solving the polynomial fitting problem:
\begin{theorem}
\label{t:mainThm}
Let $f$ be an unknown $d$ total degree polynomial with two variables, such that $f(x,y)\in[-1,1]$ when $x,y\in[-1,1]$. Given a noise parameter $\delta>0$, a threshold $\rho>0$, a constant $c>0$ (dependent on the dimension of the data) and a sample $S$ of $O(\frac{d^2}{\delta}log(\frac{d}{\delta}))$ random points $x_i,y_i,z_i \in [-1,1]$ such that $z_i\in[f(x_i,y_i)-\delta,f(x_i,y_i)+\delta]$ for at least $\rho$ fraction of $S$. With probability at least $\frac{1}{2}$ (over the choice of $S$), any feasible solution $p$ to the above LP is $c\delta$-approximation of $f$.
\end{theorem}

\begin{proof}
For our proof, we need Bernstein-Markov inequality which we state below.

\begin{theorem}(Bernstein-Markov~\cite{Ditzian1992})
\label{t:BM}
For a polynomial $P_d$ of total degree $d$, a direction $\xi$ and a bounded convex set $A \subset R^k$
\begin{center}
\begin{eqnarray}
\left\|\displaystyle\frac{\partial}{\partial\xi}P_d\right\|_\infty \leq c_{\text{\tiny{A}}}d^2\left\|P_d\right\|_\infty
\end{eqnarray}
\end{center}
where $c_{\text{\tiny{A}}}$ is independent of $d$ (and dependent on the geometric structure of A). 
\end{theorem}

Let $p=p(x,y)$ be the $d=n+m$ total degree polynomial obtained from taking any solution to the above LP.
We know $p$ exists, i.e., the LP is feasible because the coefficient of $f$ satisfies it. 

Note that although the LP constraint that 
	\[f-\delta \leq p \leq f+\delta
\]
it is NOT immediate that 
	\[\left\|f-p\right\|_\infty\leq \delta
\]
(i.e., $p$ is the $\delta$-approximation of $f$) since eq.~\ref{eq:lp1} stands for the discrete sample of $S$, where here we need to prove the $\delta$-approximation in the continuous state.

\begin{claim}
\label{c:lp1} $\left\|p\right\|_\infty \leq 1 + O(\frac{n^3m+m^3n}{|I|})$.  \end{claim}
Since $\left\|T_i(x)\right\|_\infty=\left\|T_j(y)\right\|_\infty=1$, from Bernstein-Markov (Theorem~\ref{t:BM}), we get $|T'_i(x)|=O(i^2)$. Thus:
\begin{center}
\begin{eqnarray}
\label{eq:Odiv}
 \left|p'_x\right| \leq \sum_i^{n}{\sum_j^{m}|{a_{ij}|\Big(T_i(x)T_j(y)}}\Big)'_x \nonumber\\   \leq\sum_i^{n}{\sum_j^{m}\sqrt{2}O(i^2)} \leq O(n^3m)  \nonumber\\
\end{eqnarray}
\end{center}
From symmetric consideration, we get:
\begin{center}
\begin{eqnarray}
\left|p'_{y}\right| \leq O(nm^3) \nonumber\\
\end{eqnarray}
\end{center}
By construction, $p$ takes all values in $[-1,1]$ for all points in $I$, and the distance (in $x$ direction or $y$ direction) between successive points of $I$ is $2/|I|$ ($I$ is equidistant). The claim follows from the fact that the derivative $p'$ by definition gives the rate of change in $p$: 
$\left\|p\right\|_\infty$ is between two successive points of $I$ that their values $=1$ where the possible change at the interval of length $2/|I|$ is $p'=O(n^3m+m^3n)$. 

\begin{claim}
\label{c:lp2} 
$\left\|p'_x\right\|_\infty, \left\|p'_y\right\|_\infty \leq O((n+m)^2)$. 
\end{claim}
This follows from Bernstein-Markov (Theorem~\ref{t:BM}) and the estimate $\left\|p\right\|_\infty =1+O(1)$.

Let $\epsilon$ denote the largest distance between two successive points out of $(x_1,y_1), . . . ,(x_{|S|},y_{|S|})$. Every interval of size $\epsilon$ contains at least one of the datapoints (forming $\epsilon$-net). With high probability, $\epsilon=O(log|S|/|S|)=O(\frac{\delta}{(n+m)^2})$. Now, $p$ and $f$ are functions satisfying $\left\|p'_x\right\|_\infty,\left\|f'_x\right\|_\infty, \left\|p'_y\right\|_\infty,\left\|f'_y\right\|_\infty \leq O((n+m)^2)$; hence, $\left\|(p-f)'_x\right\|_\infty, \left\|(p-f)'_y\right\|_\infty \leq O((n+m)^2)$.
Due to the LP constraint, $p,f$ differs by at most $\delta$ on the points in the $\epsilon$-net, so we get
\begin{center}
\begin{eqnarray}
 \left\|(p-f)\right\|_\infty \leq 2\delta+O(\epsilon(n+m)^2)=c\delta
\end{eqnarray}
\end{center}
which is the finished proof of Theorem~\ref{t:mainThm};
\end{proof}

{\it Remarks:}
\begin{itemize}
\item If we know that the derivative is bounded by $\Delta$ (i.e., $f'_x,f'_y \leq \Delta$), the above proof that gives us $\frac{\Delta}{\delta}{log \frac{\Delta}{\delta}}$ points is sufficient.
\item The Bernstein-Markov Theorem~\ref{t:BM} also holds for multivariate trigonometric polynomials (see~\cite{Ditzian1992}), thus, we can generalize the above proof also for this class of function. This generalization is important in the scope of wireless sensor networks since the use of trigonometric function is the appropriate way to represent the sensor data behavior (e.g., temperature).
\item The presented method holds only when we assume equidistance or random sampling (as opposed to Section~\ref{s:DiscreteFiniteNoise} that handles any given sample). Otherwise, when the dataset is dense, since we allow $\delta$ perturbation of the data, it can cause a sharp slope in the resulting function although the original data is close to the constant at the sampling interval.
\end{itemize}

\begin{corollary}
Given the set $S$ of $O(\frac{d^2}{\delta}log(\frac{d}{\delta}))$ $k$-dimensional random datapoints and a constant $c(S,k)$ dependent only on the geometry and the dimension of the data, we can reconstruct the unknown polynomial within $c(S,k)\delta$ error in $\ell_\infty$ norm with high probability over the choice of the sample.
\end{corollary}

\begin{proof}
The two-dimensional proof holds for the general dimension, where the approximation accuracy dependent on the constant $c(S,k)$ comes from Theorem~\ref{t:BM}. This constant is independent of the polynomial degree, but dependent on the set of the data points (see~\cite{Ditzian1992}). Note that $c(S,k)$ increases exponentially when increasing the dimension. 
\end{proof}

\subsection*{Byzantine Elimination.}
\label{s:ByzantineElimination}
Arora and Khot~\cite{AroraKhot2002} do not deal with Byzantine inputs; however, the method they presented in Section $6$ can be rewritten to eliminate corrupted data such that the input datapoints will contain only true (but noisy) values. 

Assume that $\rho$ fraction of the data is uncorrupted. For any point $x_i,y_i \in [-1,1]$, consider a small square-interval $\Lambda=[x_i-\frac{\delta}{d^3},x_i+\frac{\delta}{d^3}]\times[y_i-\frac{\delta}{d^3},y_i+\frac{\delta}{d^3}]$ (where $d$ is the total degree of the polynomial we need to find). For a sample of $d^4\frac{log(1/\delta)}{\delta}$ points, with high probability $\Omega(log (d))$ of the samples lie in this square. We are given that $\rho$ fraction of these sample points gives an approximate value of $f(x_i,y_i)$, i.e., the correct value lies in the interval $[f(x_i,y_i)-\delta,f(x_i,y_i)+\delta]$ and the rest of the sample is corrupted and, thus, is NOT in $[f(x_i,y_i)-\delta,f(x_i,y_i)+\delta]$. As shown in Claim~\ref{c:lp2}, the derivatives are bounded by $O(d^2)$; thus, the value of the polynomial is essentially \textit{constant} over $\Lambda$. Hence, at least $\rho$ fraction of the values seen in this square will lie in  $[f(x_i,y_i)-\delta,f(x_i,y_i)+\delta]$ and the rest is irrelevant corrupted data.
Thus, at every point $(x_i,y_i)$, we can reconstruct $f(x_i,y_i)$. The sample is large enough so that we can reconstruct the values of the polynomial at say, $d^2/\delta$ equally spaced points. Now, applying the techniques presented in Section~\ref{s:MultidimensionalDataFitting} enables us to recover the polynomial.

\subsection*{Reconstructing the Multivariate Polynomial.}
To conclude this section, we summarize the presented results in Algorithm~\ref{algo:3}:
\begin{algorithm}
\caption{Reconstruct the polynomial $p(x,y)$ representing the true data }
\begin{algorithmic}
\label{algo:3}
\REQUIRE $S,\rho,d,\delta$
\STATE $S' \leftarrow \emptyset$
\STATE $i \leftarrow 1$
\REPEAT
\STATE $\Lambda=[x_i-\frac{\delta}{d^3},x_i+\frac{\delta}{d^3}]\times[y_i-\frac{\delta}{d^3},y_i+\frac{\delta}{d^3}]$
\STATE $c \leftarrow \frac{z_1+...+z_k}{k},z_j:(x_j,y_i)\in\Lambda$
\STATE $S' \leftarrow \left\{(x_j,y_j,z_j)|(x_j,y_j)\in \Lambda \wedge z_j\approx c\right\}$
\STATE $i \leftarrow  i+1 $
\UNTIL{$|S'|> \frac{d^2}{\delta}$}
\STATE $p(x,y) \leftarrow $LP minimization (Equations~\ref{eq:lp1}-\ref{eq:lp3}) on the set $S'$
\RETURN $p(x,y)$
\end{algorithmic}
\end{algorithm}
The algorithm requires the dataset $S$, the true-data fraction $\rho$, the total degree of the expected polynomial $d$ and the noise parameter $\delta$.
In the first phase, we eliminate the Byzantine occurrence, as described in the former subsection. Assuming the given data lie in $[-1,1]\times[-1,1]$ (or translate to that interval), for the points in $S$, we are looking at the $\frac{\delta}{d^3}$-close interval and choose all the points that have {\it constant} value at this interval (this is done by the average operation). We repeat this process until we collect enough true-datapoints, i.e., at least $\frac{d^2}{\delta}$ points. This set (sign as $S'$ in the algorithm) is the input for the linear-programming equations which finally give us the expected polynomial as proof at Theorem~\ref{t:mainThm}.

\Section{Conclusions}
\label{s:Conclusion}
We have presented the concept of data interpolation in the scope of sensor data aggregation
and representation, as well as the new big data challenge, where abstraction of the data is essential in
order to understand the semantics and usefulness of the data. Interestingly, we found that classical
techniques used in numeric analysis and function approximation, such as the Welsh-Berlekamp
efficient removal of corrupted data, Arora Khot and the like, relate to the data interpolation problem.
Since the sensor aggregation task is usually a collection of inputs from \textit{spatial} sensors, for the first time we have extended existing classical techniques for the case of three or even more function dimensions, finding polynomials that approximate the data in the presence of noise and limited portion of completely corrupted data. 

We believe that the mathematical techniques we have presented have
applications beyond the scope of sensor data collection or big data, in addition to being an
interesting problem that lies between the fields of error-correcting
and the classical theory of approximation and curve fitting.

Throughout the research we have distinguished two different measures
for the polynomial fitting to the Byzantine noisy data problem: the first
being the Welsh-Berlekamp generalization for discrete-noise
multidimensional data and the second being the linear-programming
evaluation for multivariate polynomials.

Approached by the error-correcting code methods, we have suggested a way to
represent a noisy-malicious input with a multivariate polynomial. This
method assumes that the noise is discrete. When the noise is unrestricted, based on Bernstein-Markov Theorem and Arora \& Khot algorithm, we have suggested a method to reconstruct algebraic or trigonometric polynomial that traverses $\rho$ fraction of the the noisy multidimensional data.

We suggest to use polynomial to represent the abstract data since polynomial is dense in the function space on bounded domains (i.e., they can approximate other functions arbitrarily well) and have a simple and compact representation as oppose to spline e.g.,~\cite{spline1989} or others image processing methods.

Directions for further investigation might include the use of
interval computation for representing the noisy data with interval
polynomials.


%



%

\appendix
\section{Appendix}
\label{appendixA}
Given that the goal unknown polynomial $p$ has $m=2$ variable, $deg(p)=1$ and that the data contain $t=2$ Byzantine appearance, we can define the error-correcting polynomial $e$ to be univariate polynomial, $deg(e)=2$ and get the linear equation:
\begin{equation}
\label{eq:bw2Ex}
\alpha_1x^3+\alpha_2x^2y+ \alpha_3xy^2+\alpha_4y^3+\alpha_5x^2+\alpha_6xy+\alpha_7y^2+\alpha_8x+\alpha_9y+\alpha_{10}=z(x^2+\beta_{1}x+\beta_{2})
\end{equation}
when substitute the given data
$$
\text{\small (1,2,2),(-2,6,0;),(2,2,4),(6,1,7),(4,3,7),(9,1,10),(3,7,10),(5,7,12),(7,4,11),(10,3,13),(11,2,13),(12,4,16)}
$$
at (eq.\ref{eq:bw2Ex}) we get:
\begin{eqnarray}
q_1(x,y)=xy - 2y - 2x + x^2y + x^2 + x^3 \nonumber\\
e_1(x)= x^2 + x - 2\nonumber
\end{eqnarray}
Or, by defining $e$ to be bivariate polynomial,$deg(e)=1$:
\begin{equation}
\alpha_1x^3+\alpha_2x^2y+ \alpha_3xy^2+\alpha_4y^3+\alpha_5x^2+\alpha_6xy+\alpha_7y^2+\alpha_8x+\alpha_9y+\alpha_{10}=z(x+\beta_{1}y+\beta_{2})
\end{equation}
which its solution is:
\begin{eqnarray}
q_2(x,y)=x^2 + 7xy/4 - 5x/2 + 3y^2/4 - 5y/2 \nonumber\\
e_2(x,y)=x + 3y/4 - 5/2\nonumber
\end{eqnarray}
At both cases the polynomial devision result is equals and gives the expected solution:
\begin{equation}
q_1(x,y)/e_1(x,y)=q_2(x,y)/e_2(x,y)= x+y = p(x).
\end{equation}

\end{document}